\documentclass[aps,prl,reprint,amsmath,amssymb,superscriptaddress,floatfix,longbibliography]{revtex4-2}
\usepackage{amsthm,amsfonts,ctable,url,braket,mathtools}
\newtheorem{theorem}{Theorem}
\newtheorem{prop}{Proposition}
\newtheorem{defn}{Definition}
\newtheorem{example}{Example}
\newtheorem{corollary}{Corollary}
\newcommand\mc{\mathcal}
\usepackage{natbib}
\usepackage[colorlinks,
            linkcolor=blue,
            anchorcolor=blue,
            citecolor=blue,
urlcolor=blue
            ]{hyperref}
\usepackage{graphicx}
\usepackage[capitalise]{cleveref}
\crefalias{prop}{proposition}
\crefalias{defn}{Definition}

\usepackage{caption}
\usepackage{subcaption}

\usepackage[T1]{fontenc}
\usepackage{lmodern}
\newcommand{\Tr}{\operatorname{Tr}}
\newcommand{\tr}{\operatorname{Tr}}
\newcommand{\Range}{\operatorname{Range}}
\newcommand{\sign}{\operatorname{sign}}

\newcommand{\Conv}{\operatorname{Conv}}
\DeclarePairedDelimiter\norm{\lVert}{\rVert}

\begin{document}

\title{Quantum Maps Between CPTP and HPTP}

\author{Ningping Cao}
\email{ncao@uwaterloo.ca}
\affiliation{Institute for Quantum Computing, University of Waterloo, Waterloo N2L 3G1, Ontario, Canada}
\affiliation{Perimeter Institute, Waterloo N2L 2Y5, Ontario, Canada}

\author{Maxwell Fitzsimmons}
\email{mfitzsim@uwaterloo.ca}
\affiliation{Department of Applied Mathematics, University of Waterloo N2L 3G1, Ontario, Canada}

\author{Zachary Mann}
\email{zmann@uwaterloo.ca}
\affiliation{Institute for Quantum Computing, University of Waterloo, Waterloo N2L 3G1, Ontario, Canada}

\author{Rajesh Pereira}
\email{pereirar@uoguelph.ca}
\affiliation{Department of Mathematics $\&$ Statistics, University of Guelph, Guelph N1G 2W1, Ontario, Canada}

\author{Raymond Laflamme}
\email{laflamme@uwaterloo.ca}
\affiliation{Institute for Quantum Computing, University of Waterloo, Waterloo N2L 3G1, Ontario, Canada}
\affiliation{Perimeter Institute, Waterloo N2L 2Y5, Ontario, Canada}

\date{\today}

\begin{abstract}
For an open quantum system to evolve under CPTP maps, assumptions are made on the initial correlations between the system and the environment. Hermitian-preserving trace-preserving (HPTP) maps are considered as the local dynamic maps beyond CPTP. In this paper, we provide a succinct answer to the question of what physical maps are in the HPTP realm by two approaches. The first is by taking one step out of the CPTP set, which provides us with Semi-Positivity (SP) TP maps. The second way is by examining the physicality of HPTP maps, which leads to Semi-Nonnegative (SN) TP maps.
Physical interpretations and geometrical structures are studied for these maps. The non-CP SPTP maps $\Psi$ correspond to the quantum non-Markovian process under the CP-divisibility definition ($\Psi = \Xi \circ \Phi^{-1}$, where $\Xi$ and $\Phi$ are CPTP). When removing the invertibility assumption on $\Phi$, we land in the set of SNTP maps.
A by-product of set relations is an answer to the following question -- what kind of dynamics the system will go through when the previous dynamic $\Phi$ is non-invertible. In this case, the only locally well-defined maps are in $\mathcal{SN}\backslash \mathcal{SP}$, they live on the boundary of $\mathcal{SN}$. Otherwise, the non-local information will be irreplaceable in the system's dynamic.

With the understanding of physical maps beyond CPTP, we prove that the current quantum error correction scheme is still sufficient to correct quantum non-Markovian errors.
In some special cases, lack of complete positivity could provide us with more error correction methods with less overhead.
\end{abstract}
\maketitle

\section{Introduction}

Quantum channels, also known as completely positive trace-preserving (CPTP) maps, play a critical role in almost all aspects of quantum information. Non-unitary CPTP maps characterize open system dynamics under certain assumptions on the initial correlation between the environment $E$ and the system $S$~\cite{carteret2008dynamics,shabani2009vanishing,brodutch2013vanishing,dominy2016beyond}. 

In the discussion of whether the reduced dynamic of an open system $S$ is CPTP, locality is critical.
Local, in this context, means there exists a map $\Psi$ that could determine the future evolution by only taking in the local density matrix $\rho_S$.
Looking at the system $S$ from the whole system $SE$, the dynamic $\Psi$ of $S$,
\begin{equation*}
    \Psi(\rho_{SE}) = \Tr_E(U\rho_{SE}U^\dag),
\end{equation*}
is necessarily CPTP since the action of partial trace $E$ is CPTP. 
Therefore, open-system dynamics beyond CPTP is only valuable when one does not have complete information on both $S$ and $E$. The exploration in this paper considers only local maps. 

When one only focuses on the system $S$, a sufficient condition for the open system dynamic to be CP is as follows. 
If the input state of $S$ is separable from the environment ($\rho_S\otimes\rho_E$) and the initial state $\rho_E$ of the environment is fixed,
\begin{equation}\label{eq:CPTP}
    \Psi(\rho_S) = \Tr_E[U(\rho_S\otimes\rho_E) U^\dag],
\end{equation}
then the evolution of the given system is CPTP~\cite{nielsen2010quantum}.
Generally, if backtracking in time, one may find the point when the system and environment are not yet correlated. Therefore, the system's evolution afterwards can be characterized by CPTP maps. 
However, that particular starting point may not be accessible or interesting to us. 
When this separation condition is not satisfied, there is no guarantee that the system's dynamics could be described as a CPTP map or, in the worst scenario, not even a well-defined linear map that is local to the system.

There are many attempts to go beyond CPTP dynamics and interpret the meaning of non-CP open system dynamics~\cite{pechukas1994reduced,alicki1995comment,shaji2005who,jordan2004dynamics,salgado2004evolution,carteret2008dynamics,dominy2016beyond}.
An initial correlation between $S$ and $E$ may cause the subsystem dynamic $\Psi$ to be non-CP. More general maps such as assignment maps~\cite{pechukas1994reduced,carteret2008dynamics,brodutch2013vanishing}, $\mathbb{C}$-linear HPTP maps~\cite{dominy2016beyond} are proposed. 
Positive trace-preserving (PTP)~\cite{shaji2005who} and Hermitian-preserving trace-preserving (HPTP)~\cite{breuer2016colloquium,watrous2018theory,dominy2016beyond,shabani2009vanishing} are well known in quantum information and mathematical literature. 

What are local physical maps? Or what kind of map can describe a dynamic process in physical systems?
In this paper, we provide an answer to this question under the following two assumptions.
First, the Hamiltonian $H_{SE}$ at different time periods can vary, meaning that the whole process may not be governed by one master equation, which is more suitable for quantum information and quantum computing. Second, the correlation between $S$ and $E$ at the starting point of $\Psi$ is generated by a unitary $U_c$ (referred to as the \textit{context unitary} since it builds the context for later evolution), i.e. they are separable before $U_c$.

We approach this question from two directions: taking one step beyond CPTP maps, and examining the physicality of HPTP maps. 
Stepping out of the CPTP realm, we discover semi-positive trace-preserving (SPTP) maps.
By examining the physicality of HPTP maps, we bring semi-nonnegative trace-preserving (SNTP) maps to light.
After a meticulous examination of the maps between CPTP and HPTP, we prove the following hierarchy
\begin{equation*}
    \small{\text{CPTP} \subseteq \text{PTP} \subseteq \text{SPTP (after reduction)} \subseteq \text{SNTP} \subseteq \text{HPTP}.}
\end{equation*}
The physical meanings and geometric structures of these classes of maps are studied. 
SPTP maps correspond to quantum non-Markovian process under the CP-divisibility definition ($\Psi = \Xi \circ \Phi^{-1}$, where $\Xi$ and $\Phi$ are CPTP)~\cite{rivas2014quantum,breuer2016colloquium}.
When the invertibility assumption on the preconditioning process $\Phi$ is removed, we land in the set of SNTP maps.
As illustrated in~\cref{fig:VennDiagram}, the sets of SPTP and SNTP maps are star-shaped, and positive maps live in their star center. 
The closure of the SPTP maps is the SNTP maps. The convex hulls of the SPTP and SNTP maps are the HPTP maps.

CPTP is context-independent while building connections between $S$ and $E$ is important for SNTP and SPTP.
The context-independent property makes CPTP maps unique from an operational perspective -- there is no need to build an initial correlation between the system and an ancillary qubit before implementing these operations.

\begin{figure}[ht]
    \includegraphics[width = .8\linewidth]{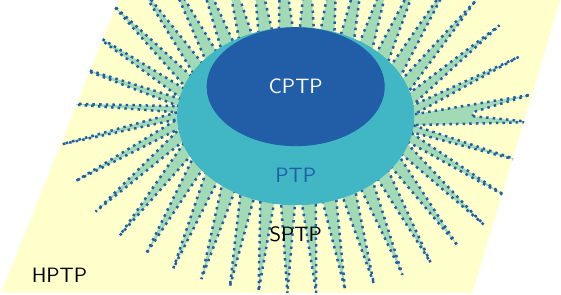}
    \caption{\textbf{Schematic diagram of map hierarchy}: Not all hermitian-preserving trace-preserving (HPTP) maps can be found in physical systems. While CPTP maps are context-independent, the non-CP SPTP and SNTP maps are context-dependent. CPTP maps and PTP maps are convex and compact. On the contrary, the sets of SPTP and SNTP maps are star-shaped and unbounded. The closure of the SPTP set is the set of SNTP.}\label{fig:VennDiagram}
\end{figure}

Since we know that the local dynamics of a system is richer than CPTP, what if the noise in the system is non-CP, how does it affect QEC?
With the physical understanding of maps beyond CPTP, we prove that the current quantum error correction scheme is still sufficient to correct errors beyond CPTP.
In some special cases, lack of complete positivity could provide us with more error correction methods with less overhead.

This paper is organized as follows. We first introduce the motivation, definitions, characterizations, and physical interpretations of SPTP and SNTP maps in Sec II. The geometric results and relations between sets of maps are studied in Sec III. In Sec IV, we prove that the Quantum Error Correction Criteria still hold for non-Markovian noise. 
Discussion and open problems are posed in the last section.

\section{Physical Non-CP Maps and Where to Find Them}

\subsection{Beyond Completely Positivity: Semi-Postive}

\begin{figure}[ht]
    \centering
    \includegraphics[width = .55\linewidth]{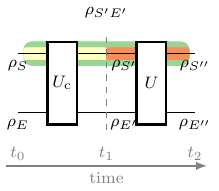}
    \caption{The unitary $U_c$ is called the context unitary since it sets the context for the evolution starting from $t_1$. Denote the processes from $t_0$ to $t_2$ (green) and $t_0$ to $t_1$ (yellow)  for the system $S$ as $\Xi$ and $\Phi$, respectively. The dynamics from $t_1$ to $t_2$ is characterized by $\Psi := \Xi\circ \Phi^{-1}$ given $\Phi$ is invertible.}
    \label{fig:2_unitary}
\end{figure}

The first step in relaxing the requirements for complete positivity is to create an initial correlation between the system and the environment. As shown in~\cref{fig:2_unitary}, we do so by introducing a context-unitary $U_c$, then study the process $\Psi$ that maps $\rho_{S'}$ to $\rho_{S''}$ (orange colored in~\cref{fig:2_unitary}). 
Assume the system $S$ and environment $E$ are first in a separable state $\rho_{SE} = \rho_S\otimes \rho_E$. After the context-unitary $U_c$, $S$ and $E$ may no longer be separable.
Denote the evolution of system $S$ undergoing $U_c$ and $U\circ U_c$ as $\Phi$ and $\Xi$ respectively. 
Clearly, the maps $\Phi$ and $\Xi$ are CPTP.

If there exists a well-defined map $\Psi: \rho_{S'} \mapsto \rho_{S''}$ that represents the dynamics from $t_1$ to $t_2$, we have
    \begin{align}
        \rho_{S''} & = \Psi(\rho_{S'}) = \Psi(\Tr_E(\rho_{S'E'})) = \Psi(\Tr_E(U_c(\rho_{SE})U_c^\dag)) \nonumber\\ 
        &= \Psi\circ \Phi(\rho_{S}) = \Xi(\rho_{S}). \label{eq:xi_phi}
    \end{align}
When $\Phi$ is invertible, from~\cref{eq:xi_phi},
the evolution $\Psi: \rho_{S'} \mapsto \rho_{S''}$ is mathematically equivalent to 
\begin{equation}\label{eq:SP}
\Psi := \Xi\circ \Phi^{-1}.	
\end{equation}
The inverse $\Phi^{-1}$ is HPTP but not CP essentially when $\Phi$ is not unitary~\cite{nayak2006invertible}.
The composition of a CPTP map $\Xi$ and an HPTP map $\Phi^{-1}$ could give rise to another (non-CP) HPTP map $\Psi$.
Note that $\Psi$ does not even need to be positive. The input state $\rho_{S'}$ of $\Psi$ is subject to $\Phi$. Mapping a subset of all density matrices to another subset of density matrices is weaker than positivity, allowing $\Psi$ to have many novel properties compared to positive maps.
Similar structures of \cref{eq:SP} has been studied in the context of quantum Markovianity. 
The non-CP $\Psi$ is, in fact, non-Markovian under the CP-divisibility definition~\cite{rivas2014quantum,breuer2016colloquium}. We will discuss this at the end of this subsection. 

When $\Phi$ is non-invertible, there may not exist a well-defined local map describing the dynamics from $\rho_{S'}$ to $\rho_{S''}$. The information that is non-local to $S$ can be irreplaceable in $S$'s evolution. More discussion and an example can be found in Section I in the supplemental material.

In the following, we fully characterize the structure of~\cref{eq:SP} by introducing the notion of a semi-positive linear map, which is inspired by the concept of semi-positive matrices~\cite{tsatsomeros2016geometric,sivakumar2018semipositive,dorsey2016new}.

\begin{defn}\label{def:SP}
Let $\mathcal{H}$ and $\mathcal{K}$ be finite-dimensional Hilbert spaces.  
An HP map $\Psi: B(\mathcal{H}) \to B(\mathcal{K})$ is said to be semi-positive (SP) if there exists an invertible density matrix $\rho$ such that $\Psi(\rho)$ is also an invertible density matrix.
The map $\Psi$ is said to be semi-positive after reduction (SPR) if $\Psi$ as a map into $B(\mathcal K_\Psi)$ is SP, where $\mathcal{K}_{\Psi}=\operatorname{Span}(\bigcup_{a\in B(\mathcal{H})} \operatorname{Range} (\Psi(a)))$.
\end{defn}

SPR is proposed to incorporate the detail below. 
The notion of invertibility on $\Psi(\rho)$ of SP in Definition~\ref{def:SP} can be subtle in certain cases. Consider the replacement channel $\Psi_\text{rp}: \rho \mapsto \tr(\rho)\ket{\psi}\bra{\psi}$ sends every density matrix to a pure state $\ket{\psi}$. Since $\ket{\psi}\bra{\psi}$ is generally not invertible in $B(\mathcal{K})$ of an arbitrary Hilbert space $\mathcal{K}$, this map is CPTP but not SPTP. However, if $\mathcal{K}$ is restricted to be $\text{Span}_{a\in B(\mathcal{H})} \Range(\Psi_\text{rp}(a))$, the state $\ket{\psi}\bra{\psi}$ becomes invertible in $\mathcal{K}$. 
Thus, $\Psi_\text{rp}$ is SPRTP. We then include all CPTP maps as a subset of SPRTP maps in such a manner.

The following result can be viewed as the quantum version of \cite[Theorem 3.1]{sivakumar2018semipositive}.

\begin{theorem}\label{thm:semi-postive}
	Let $\mathcal{H}$ and $\mathcal{K}$ be finite dimensional Hilbert spaces and let $\Psi$ be an HPTP map from $B(\mathcal{H})$ to $B(\mathcal{K})$.
	Then TFAE
		\begin{enumerate}
			\item $\Psi$ is SP \textup{[}SPR\textup{]}.
			\item There exist CPTP maps $\Phi:B(\mathcal{H}) \to B(\mathcal{H})$ and $\Xi:B(\mathcal{H}) \to B(\mathcal{K})$ \textup{[}$\Xi:B(\mathcal{H}) \to B(\mathcal{K}_{\Psi})$\textup{]} such that $\Psi= \Xi \circ \Phi^{-1}$.
		\end{enumerate}
\end{theorem}

\begin{proof}
		$(1) \implies (2)$: Let $\rho$ be a density matrix such that $\Psi(\rho)\in B(\mathcal{K})$ [$\Psi(\rho)\in B(\mathcal{K}_\Psi)$] is an invertible density matrix. 
		Let $\Xi_\lambda(x)=\lambda \Psi(x) +(1-\lambda)\tr(x)\Psi(\rho)=\Psi(\lambda x+(1-\lambda)\tr(x)\rho)$.  
		Taking the Choi matrices of each side we get $J(\Xi_\lambda)=\lambda J(\Psi)+(1-\lambda)I_{\mathcal{H}}\otimes \Psi(\rho)$. 
		Since $J(\Psi)$ is Hermitian and $I_{\mathcal{H}}\otimes \Psi(\rho)$ is positive definite, there exists $\epsilon>0$ such that if $0<\lambda<\epsilon$ then $J(\Xi_\lambda)$ is positive semidefinite and hence $\Xi_\lambda$ is completely positive. 
		Now fix $\lambda \in (0,\epsilon)$, set $\Xi=\Xi_\lambda$ and $\Phi(x)=\lambda x+(1-\lambda)\tr(x)\rho$.  Then  $\Xi(x)=\Psi(\Phi(x))$ and hence $\Psi= \Xi \circ \Phi^{-1}$.

		$(2) \implies (1)$: Let $x$ be an invertible density matrix and let $\rho=\Phi(x)$.  Since $\Phi$ is a surjective CPTP map, $\rho$ is an invertible density matrix, and $\Psi(\rho)=\Xi(x)$ is a density matrix.
		Let $a\in B(\mathcal{H})$ with $\Xi(a)$ being invertible in $B(\mathcal K)$ [$B(\mathcal{K}_\Psi)$], then there exists a positive $k$ such that $a^*a\le kx$ and hence $\Xi(a)^*\Xi(a)\le \Xi(a^*a)\le k \Xi(x)$ with the first inequality being the celebrated Schwarz inequality \cite[Corollary 2.8]{choi1974schwarz}. Thus, $\Psi(\rho)=\Xi(x)$ is invertible in $B(\mathcal K)$ [$B(\mathcal{K}_\Psi)$].
		
\end{proof}

For a given pair $\{\Phi, \Xi\}$, the map $\Psi$ is uniquely defined by~\cref{eq:SP}. 
Yet the choices of $\rho$ and $\lambda$ in the proof are not unique, which leads to different pairs of $\{\Phi, \Xi\}$.
That is to say, the same $\Psi$ can appear in many distinct physical settings. 
The output of $\Phi$ feeds into $\Psi$, the effective input of $\Psi$ varies from the choices of $\Phi$ and $\Xi$.
The proof for~\cref{thm:semi-postive} provides a method to realize a SPTP map $\Psi$ in a physical device. Given $\Psi$, finding the decomposition $\Xi$ and $\Phi$, dilating them to unitary operations $U_\Xi$ and $U_\Phi$, letting $U_c = U_\Phi$ and $U = U_\Xi U_\Phi^\dag$, the evolution of $S$ from $t_1$ to $t_2$ is then characterized by $\Psi$.

\begin{example}\label{eg:transpose_map}(The transpose map)
    The transpose map $\Psi: \rho \mapsto \rho^T$ is known to be positive but not completely positive.
    Consider the transpose map on a single qubit.
    Let $\Xi(x)=\lambda T(x)+(1-\lambda)\Tr(x)\frac{1}{2}I$ and $\Phi(x)=\lambda x+(1-\lambda)\Tr(x)\frac{1}{2}I$ where $T$ denotes the transpose map, $I$ is the single qubit identity.
    These maps follow the construction outlined in the proof of \cref{thm:semi-postive} with the choice of $\rho=\frac{1}{2}I$ the fully mixed state. Letting $0<\lambda\leq\frac{1}{3}$ guarantees that $\Phi$ and $\Xi$ are CPTP. 
    The transpose map is then given by $T=\Xi\circ \Phi^{-1}$. These maps can be dilated to unitary matrices, allowing to represent the transpose map as a quantum circuit with the same form as \cref{fig:2_unitary}. More details on the construction of the maps and unitaries are provided in the supplemental material.
    
\end{example}

We refine the HPTP results on the inverse of an invertible CPTP map $\Phi$. Note a CPTP or PTP map which is invertible is always SP.
\begin{corollary}\label{cor:inv}
    Suppose that $\Phi$ is invertible. Then, $\Phi$ is SPTP if and only if $\Phi^{-1}$ is  SPTP. In particular, when $\Phi$ is CPTP (or PTP) then $\Phi^{-1}$ is  SPTP. 
\end{corollary}

Positive maps are contractive with respect to trace norm~\cite{siudzinska2021interpolating}, i.e. $\|\Phi(x)\|_{\text{tr}} \le \|x\|_{\text{tr}}$. Thus its inverse, if it exists, is generally non-contractive, i.e. $\|\Phi^{-1}(x)\|_{\text{tr}} > \|x\|_{\text{tr}}$. From~\cref{cor:inv}, certain SPTP maps can expand. Thus, non-positive SPTP maps possess novel informatics properties.
These maps in a physical system can increase quantum state distinguishability~\cite{chakraborty2019information,siudzinska2021interpolating,breuer2016colloquium} and violate data processing inequality~\cite{dominy2016beyond}. Examples can be found in~\cite{dominy2016beyond}.

Quantum non-Markovianity defined by divisibility is in a similar notion with~\cref{eq:SP}.
We adapt the notation from~\cite{breuer2016colloquium}.
The evolution of a system between time $s$ and $t$ is given by
\begin{equation}\label{eq:non-M}
    \Phi_{t,s} = \Phi_{t}\circ\Phi_s^{-1},
\end{equation}
the whole process $\Phi_t$ is Markovian if $\Phi_{t,s}$ is CPTP for any $s\in [0,t]$.
Quantum Non-Markovian dynamics are obviously SPTP.
The case studied in~\cref{eq:SP} is broader than~\cref{eq:non-M}. 
In the discussion of Markovianity, the evolution is driven by a fixed Hamiltonian $H_{SE}$, and hence by one master equation. 
As in~\cref{fig:2_unitary}, $U$ and $U_c$ can be from two different Hamiltonians, such as in quantum computing, where the two unitaries in a circuit are realized by different pulses.
While a range of continuous time $[0,t]$ is usually considered in the study of quantum non-Markovianity, discretized time sequences $\{t_1, t_2, \cdots\}$, especially time points between each operation, are more relevant in Quantum Computation. 
Non-Markovianity of open system dynamics signals a non-trivial memory effect between the system and the environment and information backflow~\cite{breuer2016colloquium,chakraborty2019information}.

\subsection{Examine HPTP: Semi-Nonnegative}

\begin{figure}[ht]
    \includegraphics[width = .7\linewidth]{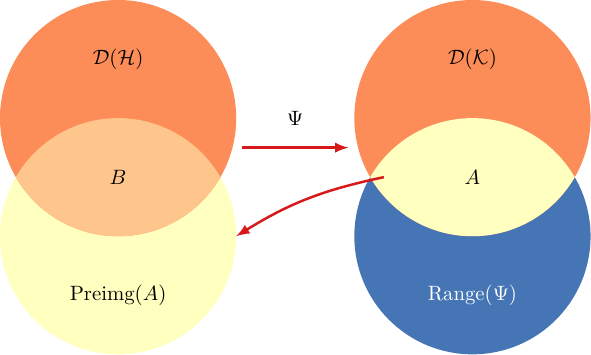}
    \caption{\textbf{$B$ determines the property of $\Psi$}: Given a HPTP map $\Psi: B(\mathcal{H}) \to B(\mathcal{K})$, the set $A$ is the intersection between $\mathcal{D}(\mathcal{K})$ and $\Range(\Psi):=\{\Psi(a) : a\in B(\mathcal{H}) \}$, where $\mathcal{D}(\mathcal{K})$ is the set of density matrices in $\mathcal{K}$. 
    The preimage under $\Psi$ of $A$ is denoted as $\text{Preimg}(A)$. The set $B$ is $\mathcal{D}(\mathcal{H})\cap \text{Preimg}(A)$. 
    (1) If $B = \mathcal{D}(\mathcal{H})$, $\Psi$ is positive; (2) if there exist a $\rho \in B$ such that $\Psi(\rho)$ is invertible, $\Psi$ is semi-positive (SP); (3) if $B\neq \emptyset$, $\Psi$ is semi-nonnegative (SN); (4) if $B = \emptyset$, $\Psi$ is a non-SN HPTP map.}\label{fig:scheme}
\end{figure}

What kind of HPTP maps are physical?
Clearly, certain HPTP maps are not interesting for quantum physics. 
A simple instance is the replacement HPTP map $\Upsilon: \rho \mapsto \tr(\rho)D$, where $D$ is an indefinite Hermitian matrix, it sends no density matrix to density matrices.
Since the partial trace is a CPTP map, an indefinite matrix can not be dilated to a density matrix in a larger Hilbert space. Thus $\Upsilon$ can never be a physical map, even considering a possible dilation. 
From this viewpoint, we require the map to send at least one density matrix to a density matrix. 
This leads to the concept of semi-nonnegative maps, where the terminology is adapted from analysis literature~\cite{dorsey2016new}.

\begin{defn}\label{def:SN} Let $\mathcal{H}$ and $\mathcal{K}$ be finite dimensional Hilbert spaces.  A HPTP map $\Psi: B(\mathcal{H}) \to B(\mathcal{K})$ is said to be semi-nonnegative (SN) if there exists a density matrix $\rho$ such that $\Psi(\rho)$ is also a density matrix.   
\end{defn}

 \begin{theorem}\label{thm:SN} 
 Let $\mathcal{H}$ and $\mathcal{K}$ be finite dimensional Hilbert spaces and let $\Psi$ be an HPTP map from $B(\mathcal{H})$ to $B(\mathcal{K})$. 
Then TFAE
 \begin{enumerate}
   \item $\Psi$ is semi-nonnegative.
   \item There exist CPTP maps $\Phi:B(\mathcal{H})\to B(\mathcal{H})$ and $\Xi:B(\mathcal{H})\to B(\mathcal{K})$ such that $\Psi\circ \Phi= \Xi$.
  \end{enumerate}
\end{theorem}

\begin{proof} 
$(1) \implies (2)$: Let $\rho$ be a density matrix such that $\Psi(\rho)$ is also a density matrix.  Let $\Phi(x)=tr(x)\rho$ and $\Xi(x)=tr(x)\Psi(\rho)$, then $\Phi$ and $\Xi$ are CPTP maps with $\Psi\circ \Phi= \Xi$.

 $(2) \implies (1)$: Now suppose there exist CPTP maps $\Phi:B(\mathcal{H})\to B(\mathcal{H})$ and $\Xi:B(\mathcal{H})\to B(\mathcal{K})$ such that $\Psi\circ \Phi= \Xi$.  Now let $\rho$ be any density matrix, then $\Psi$ maps $\Phi(\rho)$ to $\Xi(\rho)$ and hence $\Psi$ is semi-nonnegative.

\end{proof}

From~\cref{thm:SN}, a similar dilation can be constructed. The unitary operations $U_c$ and $U_\Xi$ in a larger Hilbert space represent the CPTP maps $\Phi$ and $\Xi$, respectively. 
Now it is clear that the prerequisite for SP maps and SN maps to appear in quantum systems is the context put in by the map $\Phi$, or equivalently the context unitary $U_c$. 
The CPTP map $\Phi$ restricts the density matrix sent through a non-CP map $\Psi$ to only a subset of $\mathcal{D}(\mathcal{H})$, which ensures the physicality of the results (i.e. no negative eigenvalue in $\Psi(\rho)$). 
CPTP maps are superior from the operational perspective -- they do not require an initial correlation between the system and the ancilla. 
In the landscape of HPTP maps, the CPTP maps are \textit{context-independent} physical maps, while other non-CP SNTP maps (including positive maps) are \textit{context-dependent}. The non-SN HPTP maps are never physical.

Comparing the Definition~\ref{def:SP} of SP or SPR with Definition~\ref{def:SN} of SN, the invertibility of $\rho$ is relaxed in SN. The SPTP maps are a subset of SNTP maps.
Equivalent characterizations of SP and SN are given in~\cref{thm:semi-postive} and~\cref{thm:SN}, respectively.
In the proof of~\cref{thm:SN}, we merely removed the invertible assumption for $\Phi$ from semi-positivity.
Denote the set of SNTP maps and SPTP maps as $\mathcal{SN}(\mathcal{H},\mathcal{K})$ and $\mathcal{SP}(\mathcal{H},\mathcal{K})$, respectively, or simply $\mathcal{SN}$ and $\mathcal{SP}$ when it does not cause confusion.
The elements in $\mathcal{SN}\backslash \mathcal{SP}$ have no such decomposition $\Xi\circ \Phi^{-1}$.

\begin{example}[SNTP but not SP]
    Consider the single qubit HPTP map $\Psi:$
    \[
    \begin{pmatrix}
        a & b\\
        c & d
    \end{pmatrix} \mapsto
    \begin{pmatrix}
        a+2d & b\\
        c & -d
    \end{pmatrix}.
    \]
    $\Psi$ is a non-SP SNTP map, it maps the whole Bloch sphere to indefinite matrices except $\ket{0}\bra{0}$. 

    Let $\Phi$ be the CPTP map with Kruas operators $$\left\{\begin{pmatrix}
        1 & 0\\
        0 & 0
    \end{pmatrix},\begin{pmatrix}
        0 & 1\\
        0 & 0
    \end{pmatrix}\right\}.$$ The composition $\Psi \circ \Phi(\rho) = \Phi(\rho)$ is still a CPTP map. More discussion on this map can be found in Sup. Mat.
\end{example}

From another perspective, the classes of maps can be distinguished under a unified framework.
Given a HPTP map $\Psi: B(\mathcal{H}) \to B(\mathcal{K})$, the set $A$ is the intersection between $\mathcal{D}(\mathcal{K})$ and $\Range(\Psi)$. As in \cref{fig:scheme}, the preimage of $A$ denoted as $\text{Preimg}(A)$. The set $B$ is $\mathcal{D}(\mathcal{H})\cap \text{Preimg}(A)$, it is the subset of density matrices $\mathcal{D}(\mathcal{H})$ mapped to density matrices in $\mathcal{D}(\mathcal{K})$ by $\Psi$. 

(1) If $B = \mathcal{D}(\mathcal{H})$, $\Psi$ is positive; 

(2) if there exists a $\rho \in B$ with $\Psi(\rho)$ invertible, $\Psi$ is semi-positive (see~\cref{thm:SP_noinv} in the Sup. Mat.); 

(3) if $B\neq \emptyset$, $\Psi$ is semi-nonnegative; 

(4) if $B = \emptyset$, $\Psi$ is a non-SN HPTP map. This map is non-physical. 

A by-product of this framework is a clear answer to the question: what kind of dynamics will the system go through when removing the invertibility of $\Phi$ in~\cref{eq:SP} and~\cref{eq:non-M}. The only non-SP physically interpretable maps are in the set $\mathcal{SN}\backslash \mathcal{SP}$. We later prove in~\cref{thm:SN_SP} that these maps are the boundary of $\mathcal{SN}$. Beyond the power of SNTP maps, the system $S$'s dynamics are not locally well-defined, information either stored in $E$ or stored globally will kick in and dramatically change the evolution of $S$.

\section{Relations Between Maps and Geometrical Characterization}
With the newly defined SPTP and SNTP maps in hand, it is natural to wonder where they are in the HPTP landscape and what are the relations between various sets of maps.
In this subsection, we provide the geometrical characterization of these sets.
\cref{fig:VennDiagram} is a schematic diagram of these results.
Proofs of all results in this subsection can be found in Sup. Mat.

In~\cref{thm:set_relations}, we study the geometric properties of each set of maps.
\begin{theorem}[Geometric properties of maps (informal)]\label{thm:set_relations}
     Let $\mathcal{H}$ and $\mathcal{K}$ be finite-dimensional Hilbert spaces. 
     Let $\mathcal {HP}$ be the set of all HPTP linear maps from $B(\mathcal H)$ to $B(\mathcal K)$. 
     The following statements hold:
     \begin{enumerate}
         \item The set $\mathcal {SP}$ of SPTP maps in $\mathcal {HP}$ is open and unbounded 
         but not convex. It is star-shaped with star center equal to $ \mathcal {SP}\cap \mathcal {P}$. \label{thm:set_relations.SP}
         \item The set $\mathcal {SN}$ of SNTP maps in $\mathcal {HP}$ is closed and unbounded 
         but not convex. It is star-shaped with star center equal to $ \mathcal {P}$. \label{thm:set_relations.SN}
         \item The set $\mathcal {SPR}$ of SPRTP maps in $\mathcal {HP}$ is unbounded, not open and not closed 
         nor is it convex. It is star-shaped with star center equal to $ \mathcal {P}$. \label{thm:set_relations.SPR}
         \item The set $\mathcal {P}$ of positive TP maps in $\mathcal {HP}$ is compact 
         and convex. \label{thm:set_relations.P}
         \item The set $\mathcal {CP}$ of CPTP maps in $\mathcal {HP}$ is compact 
         and convex. \label{thm:set_relations.CP}
     \end{enumerate}
\end{theorem}

\Cref{thm:set_relations.P} and \cref{thm:set_relations.CP} of~\cref{thm:set_relations} are well-known. 
The formal statements of~\cref{thm:set_relations} are provided and proved in Sup. Mat. 
In~\cref{thm:set_relations.inclusions}, we provide the inclusion relation between sets of maps. 
The map hierarchy is illustrated in~\cref{fig:VennDiagram}.
\begin{prop}[Nested structure between maps]\label{thm:set_relations.inclusions} 
Let $\mathcal{H}$ and $\mathcal{K}$ be finite-dimensional Hilbert spaces, 
the following inclusion relations hold
         \[
            \mathcal {CP}\subseteq \mathcal {P} \subseteq \mathcal{SPR} \subseteq \mathcal {SN} \subseteq \mathcal {HP}
         \]
         and
         \[
         \mathcal {SP} \subseteq \mathcal{SPR}.
         \]
The equal signs hold iff $\dim \mathcal{H} = 1$ or $\dim \mathcal{K} = 1$.
\end{prop}

It is straightforward to see that $\mathcal {SP}$ is a subset of $\mathcal {SN}$ from their definitions. The relation between $\mathcal {SN}$ and $\mathcal {SP}$ is, in fact, stronger than inclusion. 
\begin{theorem}[Relation between SP and SN]\label{thm:SN_SP}
    The interior of $\mathcal {SN}$ is $\mathcal {SP}$, the closure of $\mathcal {SP}$ is $\mathcal {SN}$, i.e.
         \[
         \begin{array}{lr}
           \operatorname{int}(\mathcal {SN})= \mathcal {SP},   &  \mathcal{SN}=\overline{\mathcal {SP}}.
         \end{array}
         \] \label{thm:set_relations.SN_regular}
\end{theorem}
This structure shows that there are not too many well-defined maps left for $\Phi$ to be non-invertible.

\begin{prop}
    The convex hulls $\Conv(\mathcal{SP})$ and $\Conv(\mathcal{SN})$ of $\mathcal{SP}$ and $\mathcal{SN}$ are the set of HPTP maps, i.e. $\Conv(\mathcal{SP}) = \Conv(\mathcal{SN}) = \mathcal{HP}$.
\end{prop}

The dual map is a natural object to consider. Mathematically, $\Psi$ is positive iff $\Psi^*$ is positive. Physically, if $\Psi$ characterizes the dynamics in the Schrodinger picture, then $\Psi^*$ represents the dynamics in the Heisenberg picture. 
The following theorem demonstrates the dual map continues to sharply characterize semi-positivity. It can be viewed as a quantum generalization of \cite[Lemma 1.5]{dorsey2016new}.

\begin{theorem}\label{thm:farkasSP}
Let $\mathcal{H}$ and $\mathcal{K}$ be finite dimensional Hilbert spaces and let $\Psi$ be an HPTP map from $B(\mathcal{H}) \to B(\mathcal{K})$.  Then, one and only one of the following statements is true:

\begin{enumerate}
   \item $\Psi$ is a semi-positive map.
   \item $-\Psi^*$ is a semi-nonnegative map, where $\Psi^*$ denotes the dual map (or adjoint) of $\Psi$.
  \end{enumerate}
\end{theorem}

\section{Application in Quantum Error Correction: Crises or Opportunities?}

Quantum error correction and fault-tolerance for non-Markovian noise have been considered in~\cite{terhal2005fault,aliferis2005quantum}. In previous studies, assumptions are made about the interaction amongst the environments of the qubits. 
From previous sections, we learnt that non-CP SPTP maps correspond to quantum non-Markovian noise without direct assumptions on interaction strength.
Can the current quantum error correction framework still correct non-Markovian noise?
In this section, we consider the noise to be an SPTP map. 
The recovery channel must still be a CPTP map since CPTP is context-independent. 

An appropriate representation for SPTP maps is necessary to formulate this problem mathematically. 
Whether a linear map $\Psi$ is Hermitian-preserving or completely positive can be easily checked by writing out the Choi representation $J(\Psi)$ of $\Psi$. If $J(\Psi)$ is a Hermitian matrix, the map $\Psi$ is HP. 
Similarly, if $J(\Psi)$ is positive semidefinite, $\Psi$ is CP~\cite{watrous2018theory}.
However, the positivity of the Choi representation does not signal a non-CP SP or non-CP SN map from a HP map. 
We provide semidefinite programming to determine the semi-positivity of a non-CP HPTP map in Sup. Mat.

Any HPTP map has an operator sum representation similar to Kraus operators $\Psi: \{\sign(i), E_i\}$.
The action of $\Psi$ is $\Psi(\rho) = \sum_i \sign(i) E_i \rho E_i^\dag$, where $\sign(i)$ is the sign function, $\sign(i) = -1, i\in\mathbf{J}$ for a non-empty subset $\mathbf{J}$ of $i$, otherwise $\sign(i) = 1$.
Discussion on representations of HPTP maps can be found in Section V of the supplemental material.
With a modification of the proof in~\cite[Theorem 10.1]{nielsen2010quantum}, we prove that the Quantum Error Correction Criteria, a.k.a. Knill-Laflamme condition~\cite{knill1997theory,knill2000theory}, is still sufficient for correcting SPTP errors (in fact, any HPTP errors).

\begin{theorem}\label{thm:sufficient}
Let $\mathbf{C}$ be a code space, $P$ be the projector onto the code space $\mathbf{C}$. The operator sum representation of the semi-positive trace-preserving (SPTP) noise map $\mc{N}$ is given by $\{\sign(i), E_i\}$, where $\sign(i)$ is the sign function of $\mc{N}$ (for a non-empty subset $\mathbf{J}$ of $i$, $\sign(i) = -1, i\in\mathbf{J}$, otherwise $\sign(i) = 1$). A sufficient condition for a CPTP recovery map $\mc{R}$ correcting $\mc{N}$ is that 
\begin{equation}\label{eq:klcondition}
    PE_iE_j^\dag P = \alpha_{ij} P
\end{equation}
where $\alpha$ is a Hermitian matrix.
\end{theorem}

\begin{proof}

WLOG, assume $\alpha$ is diagonal. 

Polar decomposition $E_k P = U_k\sqrt{P E_k^\dag E_k P} = \sqrt{\alpha_{kk}} U_k P$.
Let $P_k : = U_k P U_k^\dag = {E_k P U_k}/{\sqrt{d_{kk}}}$. Since $\alpha$ is diagonal, $P_k\cdot P_l = 0$ for $k\neq l$.
Let $R_k = U_k^\dag P_k$, we have 
\begin{align*}
\mathcal{R}\circ\mathcal{N}(\rho) &= \sum_{kl} \sign(l) U_k^\dag P_k E_l P \sigma P E_l^\dag P_k U_k\\
& = \sum_k \alpha_{kk} \sign(k) P \sigma P
\end{align*}
Semi-positive maps are also Hermitian preserving (HP).
For a trace-preserving HP map, $\sum_i \sign(i) E_i^\dag E_i = I$. 
\[
P\left(\sum_i \sign(i) E_iE_i^\dag\right) P = \sum_i \sign(i) \alpha_{ii} P
\]
Therefore, $\sum_i \alpha_{ii} \sign(i) = 1$.
\[
\mc{R}\circ\mc{N}(P\sigma P) = \sum_i \alpha_{ii} \sign(i) P\sigma P = P\sigma P
\]
\end{proof}

However, the Knill-Laflamme condition is no longer a necessary condition.
When $\mc{N}^{-1}$ is CPTP, we can choice $\mc{R} = \mc{N}^{-1}$. And $\mc{R}\circ\mc{N}(\rho) = \rho$ holds for all density matrices. In this case, the Knill-Laflamme condition fails. That means we have a new way to correct this type of error that does not involve subspace codes.
The information backflow in SPTP noise allows parts of the information to be restored without active correction. 

\section{Outlook}

\Cref{thm:semi-postive} and \Cref{thm:SN} offer ways to interpret and realize SPTP and SNTP maps in quantum systems. We shall point out that it may not be the only physical interpretation. The connection between SPTP/SNTP maps and other notions of physical maps, such as assignment maps and $\mathbb{C}$-linear HPTP maps, awaits to be studied. 

The set $B$ defined in~\cref{fig:scheme} is the set of density matrices mapped to density matrices by $\Psi$.
Although for any given $\rho\in B$, we can find a pair of CPTP maps $\{\Phi_\rho, \Xi_\rho\}$ such that $\rho$ is in the range of $\Phi_\rho$, $B$ can not be covered by one construction $\{\Phi, \Xi\}$ in general.
In our construction for~\cref{eg:transpose_map}, only $1/27$ of the Bloch sphere is covered.
What stops us from conveying more elements in $B$ in one pair of $\{\Phi, \Xi\}$? 
How can one find a maximal division for a given SPTP map $\Psi$? Many interesting questions arise here.

The single qubit non-SP SNTP maps only send one pure state to a pure state (Section III, Sup. Mat.). They map other density matrices to non-density matrices. For the physical scenarios that these maps can characterize, information backflow would not happen, and these maps are not uniquely defined.
The distinction between SN and SP becomes tricky in higher-dimensional Hilbert spaces $\mc{H}$, mainly due to the geometry of quantum states $\mathcal{D}(\mathcal{H})$. Unlike single qubit density matrices which form the Bloch sphere, a higher dimensional unit ball does not represent $\mathcal{D}(\mathcal{H})$~\cite{bengtsson2017geometry}. Non-pure, not full-rank states comprise flat facets on the surface of $\mathcal{D}$. A non-SP SNTP map can map one of those facets (or part of it) to valid quantum states.
We conjecture that non-SP SNTP maps still signal no information backflow and are not uniquely defined even for higher dimensional $\mathcal{H}$. 

It would also be interesting to study how would current results on open system simulation, noise characterization protocols etc change if relaxing a possible CPTP assumption.

\section{Acknowledgement}
NC and ZM thank Daniel Gottesman for the helpful discussion.
RP would like to acknowledge the support of Discovery grant no. RGPIN-2022-04149.
RL thanks Mike and Ophelia Lazaridis for funding.

\bibliography{ref}
\bibliographystyle{unsrt}

\newcommand{\veco}{\operatorname{vec}}
\newcommand{\rank}{\operatorname{rank}}

\newcommand{\lrp}[1]{\left( #1 \right) }
\newcommand{\lrb}[1]{\left[ #1 \right] }
\newcommand{\lrcb}[1]{\{ #1 \} }

\newtheorem{defnn}[theorem]{Definition}
\newtheorem{examp}[theorem]{Example}
\newtheorem{pro}[theorem]{Proposition}
\newtheorem{lemma}[theorem]{Lemma}
\renewcommand*{\thetheorem}{S\arabic{theorem}}

\newpage
\onecolumngrid
\begin{center}
\textbf{Supplemental Material: Quantum Maps Between CPTP and HPTP}
\end{center}

\section{Non-Invertible $\Phi$}

When $\Phi: \rho_S \mapsto \rho_{S'}$ is non-inveritible, the problem of finding a map $\Psi$ describes the process from $\rho_{S'}$ to $\rho_{S''}$ is that it may not be well-defined. We illustrate this view with the following example.

\begin{figure}[ht]
    \centering
    \includegraphics{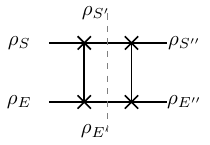}
    \caption{Two swap gates}
    \label{fig:swap}
\end{figure}
\cref{fig:swap} is an example that $\Phi$ is not invertible. All possible $\rho_S$ maps to $\rho_E$. Therefore, after the first unitary (a swap gate), the information of $\rho_S$ is completely erased from system $S$ but stored in the environment $E$. The action of the second unitary allows the information stored in $\rho_{E'}$ backflow, and the local density matrix $\rho_{S''}$ recover to $\rho_S$. 
The picture is clear when we have complete knowledge about the whole system $SE$, but if we only have access to $S$, the process from $\rho_{S'}$ and $\rho_{S''}$ can not be written as a well-defined map.

We also remark that it is possible that the information about $S$ stores not locally in $E$. It can be stored globally, shared between $S$ and $R$, and may not be seen locally in any subsystem.

\section{Unitary Operators for Realizing the Single Qubit Transpose Map}

We take $\Xi$ and $\Phi$ as defined in Example 1 of the main text. The Choi representations for $\Xi$ and $\Phi$ are given by 
\[J\lrp{\Xi}=\frac{1+\lambda}{2}\left( \ket{00}\bra{00}+\ket{11}\bra{11}+\ket{\Psi^+}\bra{\Psi^+} \right)+\frac{1-3\lambda}{2}\ket{\Psi^-}\bra{\Psi^-}\]
and
\[J\lrp{\Phi}=\frac{1-\lambda}{2}\left( \ket{01}\bra{01}+\ket{10}\bra{10}+\ket{\Phi^-}\bra{\Phi^-} \right)+\frac{1+3\lambda}{2}\ket{\Phi^+}\bra{\Phi^+}\]
respectively, where $\ket{\Psi^{\pm}}$ and $\ket{\Phi^{\pm}}$ denote the Bell states. Remark that $0<\lambda\leq\frac{1}{3}$ implies that these matrices have all non-negative eigenvalues, and thus are positive semi-definite, implying the maps $\Xi$ and $\Phi$ are completely positive. Operator-sum representations of $\Phi$ and $\Xi$ can be constructed using the operators 
\[ A_0=\sqrt{\frac{1-\lambda}{2}}\ket{0}\bra{1},\quad A_1=\sqrt{\frac{1-\lambda}{2}}\ket{1}\bra{0},\quad A_2=\frac{\sqrt{1+3\lambda}}{2}I,\quad A_3=\frac{\sqrt{1-\lambda}}{2}Z \]
for $\Phi$ and 
\[ B_0=\sqrt{\frac{1+\lambda}{2}}P_0,\quad B_1=\quad \sqrt{\frac{1+\lambda}{2}}P_1,\quad B_2= \frac{\sqrt{1+\lambda}}{2}X,\quad B_3= \frac{\sqrt{1-3\lambda}}{2}iY \]
 for $\Xi$. Here, $P_0=\ket{0}\bra{0}$, $P_1=\ket{1}\bra{1}$ are projection operators and $X,Y,Z$ denote the Pauli operators. For convenience, we shall instead use the operators 
\[C_0=A_0,\quad C_1=A_2,\quad C_2=\frac{1}{\sqrt{2}}\left(A_2+A_3\right),\quad C_3=\frac{1}{\sqrt{2}}\left( A_2-A_3 \right)\]
 for the $\Phi$ map, and the operators
\[D_0= B_0,\quad D_1= B_1,\quad D_2=\frac{1}{\sqrt{2}}\left( B_2+B_3 \right),\quad D_3=\frac{1}{\sqrt{2}}\left( B_2-B_3 \right) \]
for the $\Xi$ map. The unitary equivalence between the two representations for $\Xi$ and the two representations of $\Phi$ is obtained using the matrix $I\oplus H$ where $H$ is the Hadamard gate, $I$ is the identity matrix and $\oplus$ denotes the matrix direct sum. Unitary dilations for $\Xi$ and $\Phi$ can be obtained as follows : 
\[U_{\Phi}=\begin{pmatrix}
C_2 & -C_3 & C_0^{\dagger} & C_1^{\dagger} \\
C_3 & C_2 & C_1^{\dagger} & -C_0^{\dagger} \\
C_0 & C_1 & -C_3 & C_2 \\
C_1 & -C_0 & -C_2 & -C_3 
\end{pmatrix},\quad 
U_{\Xi}=\begin{pmatrix}
D_0 & -D_1 & D_2^{\dagger} & D_3^{\dagger} \\
D_1 & D_0 & D_3^{\dagger} & -D_2^{\dagger} \\
D_2 & D_3 & -D_1 & D_0 \\
D_3 & -D_2 & -D_0 & -D_1 
\end{pmatrix} 
\]

\section{Non-SP SNTP Maps}

\begin{examp}[SNTP but not SP]\label{eg:nonSP_app}
    Consider the single qubit HPTP map $\Psi:$
    \[
    \begin{pmatrix}
        a & b\\
        c & d
    \end{pmatrix} \mapsto
    \begin{pmatrix}
        a+2d & b\\
        c & -d
    \end{pmatrix}.
    \]
    $\Psi$ is a non-SP SNTP map, it maps the whole Bloch sphere to indefinite matrices except $\ket{0}\bra{0}$. 

    Let $\Phi$ be the CPTP map with Kruas operators $$\left\{\begin{pmatrix}
        1 & 0\\
        0 & 0
    \end{pmatrix},\begin{pmatrix}
        0 & 1\\
        0 & 0
    \end{pmatrix}\right\}.$$ The composition $\Psi \circ \Phi(\rho) = \Phi(\rho)$ is still a CPTP map.
\end{examp}

This example is provided in the main text. It is easy to see that $\Phi(\rho) = \ket{0}\bra{0}$ for any $\rho\in \mathcal{D}(\mathcal{H})$. Therefore $\Psi$ can be demonstrated in the circuit in~\cref{fig:nonSP_SN}. However, any unitary that sends $\ket{0}$ to $\ket{0}$ can replace $\Psi$ in the plot. That is to say, the process in the dashed-line box is not uniquely characterized by $\Psi$. We conjecture that this is true for all maps in $\mathcal{SN}\backslash\mathcal{SP}$.
\begin{figure}[ht]
    \centering
    \includegraphics{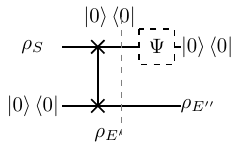}
    \caption{The SNTP map $\Phi$ in the Example~\ref{eg:nonSP_app} represents the depicted circuit.}
    \label{fig:nonSP_SN}
\end{figure}

\section{Proof for Geometric Results and Set Relations}

In this section, we provide the proofs of all statements in Sec III of the main text.
Especially, we combined Theorem 3, Proposition 2 and Theorem 4 in the main text to~\cref{thm:set_relations_app}.

\begin{lemma} \label{thm:SP_noinv}
Let $\mathcal{H}$ and $\mathcal{K}$ be finite dimensional Hilbert spaces and let $\Psi$ be an HP map from $B(\mathcal{H})$ to $B(\mathcal{K})$.
Then TFAE:
 \begin{enumerate}
   \item $\Psi$ is SP.
   \item There exists $\rho\geq 0$, $\rho \in B(\mathcal H)$ with $\Psi(\rho) >0$.
  \end{enumerate}
\end{lemma}
\begin{proof}
    Let $\mathcal S_{\mathcal K}$ be the set of invertible positive definite matrices in $B(\mathcal K)$ then $\mathcal S_{\mathcal K}$ is open in $B(\mathcal{K})$. Since $\Psi$ is continuous we have that $\Psi^{-1}(\mathcal S_{\mathcal K})$ is open in $\mathcal H$. 
    
    If 2 holds then $\rho \in \Psi^{-1}(\mathcal S_{\mathcal K})$, as this set is open there is a $\epsilon>0$ such that $\rho+\epsilon I \in \Psi^{-1}(\mathcal S_{\mathcal K})$. Since $\rho+\epsilon I >0$ holds and $\Psi(\rho+\epsilon I ) \in \mathcal S_{\mathcal K}$  we know that $\Psi$ is SP.
\end{proof}

\begin{lemma} \label{thm:range_reduection_app}
Let $\mathcal{H}$ and $\mathcal{K}$ be finite dimensional Hilbert spaces and let $\Psi$ be an HP map from $B(\mathcal{H})$ to $B(\mathcal{K})$. 
Let $\mathcal{K}_\Psi= \bigcup_{a\in B(\mathcal K)}\operatorname{Range} (\Psi(a))$. Then there is an invertible density matrix $\rho \in \mathcal H$ with $\mathcal{K}_\Psi= \operatorname{Range} (\Psi(\rho))$.
\end{lemma}
\begin{proof}
    Since $\mathcal K$ is finite dimensional there exists a $X\in B(\mathcal H)$ with $\Psi(X)$ having maximum rank in $\operatorname{Range}(\Psi)$. 
    We claim that there is a $\epsilon>0$ such that $\rho'=\epsilon X+I_{\mathcal H}>0$ has the same range as $\Psi(X)$.
    To see this, note that there is a $L>0$ such that for all $\epsilon \in (0,L)$ we have that $\epsilon X+I_{\mathcal H}>0$. 
    For all $a\in B(\mathcal K)$ define $\sigma_{\max}(a)$ to be the largest singular value of $\Psi(a)$ and let $\sigma_{\min}(a)$ be the smallest non-zero singular value of $\Psi(a)$.
    Pick $\epsilon \in (0,\min\{L,\frac{\sigma_{\min}(I_{\mathcal H})} {\sigma_{\max}(X)} \} )$, let $\ket \phi\in \mathcal K$ have $\Psi(I_{\mathcal H})\ket \phi\neq 0$ or $\Psi(X)\ket \phi\neq 0$  and for the sake of contradiction suppose that  $\Psi(\epsilon X+I_{\mathcal H})\ket \phi =0$. 
    Then, $\epsilon \Psi(X)\ket \phi =-\Psi(I_{\mathcal H})\ket \phi \neq 0$ and 
    \[
        \sigma_{\min}(I_{\mathcal H})>\epsilon \sigma_{\max}(X) \geq  \epsilon\norm{\Psi(X)\ket \phi} \geq \norm{\Psi(I_{\mathcal H})\ket \phi} \geq \sigma_{\min}(I_{\mathcal H}),
    \]
    where the norm is the norm induced by the inner product on $\mathcal K$. This shows that
    $\operatorname{Range}(\Psi(I_{\mathcal H})) \cup \operatorname{Range}(\Psi(X)) \subseteq \operatorname{Range}(\Psi(\epsilon X+I_{\mathcal H}))$
    However, $\Psi(X)$ has maximum rank so $\operatorname{Rank}(\Psi(\epsilon X+I_{\mathcal H})) = \operatorname{Rank}(\Psi(X))$, which is enough to show that $\operatorname{Range}(\Psi(\epsilon X+I_{\mathcal H})) = \operatorname{Range}(\Psi(X))$  for all $\epsilon \in (0,\min\{L,\frac{\sigma_{\min}(I_{\mathcal H})} {\sigma_{\max}(X)} \} )$. Therefore, we pick a  $\epsilon \in (0,\min\{L,\frac{\sigma_{\min}(I_{\mathcal H})} {\sigma_{\max}(X)} \} )$ and set $\rho' = \epsilon X+I_{\mathcal H}$ to prove the claim.

    Let $\rho = \frac{1}{\Tr(\rho')} \rho'$ then $\rho$ is an invertible density matrix, and from above, we know that $\Psi(\rho)$ has maximum rank. 
    It follows from definitions that $\mathcal{K}_\Psi\supseteq \operatorname{Range} (\Psi(\rho))$. 
    So consider any $a\in B(\mathcal H)$ and suppose that $\Psi(a)\ket \phi \neq 0$, by the arguments given above there is a $\epsilon'>0$ with $\operatorname{Range}(\Psi(\rho)) \cup \operatorname{Range}(\Psi(a)) \subseteq \operatorname{Range}(\Psi(\epsilon' a+\rho))$. But since $\Psi(\rho)$ has maximum rank we know that 
    $\operatorname{Range}(\Psi(\rho))  = \operatorname{Range}(\Psi(\epsilon' a+\rho)) \supseteq \operatorname{Range}(\Psi(a))$ for all $a\in B(\mathcal H)$. Hence, $\mathcal{K}_\Psi\subseteq \operatorname{Range} (\Psi(\rho))$ too. Therefore the $\mathcal{K}_\Psi=\operatorname{Range} (\Psi(\rho))$ and $\rho$ is an invertible density matrix.
\end{proof}

\begin{theorem}\label{thm:star}
	The star center of $\mathcal{SN}$ is equal to $\mathcal P$. The star center of $\mathcal{SP}$ is equal to $\mathcal P\cap \mathcal{SP}$. The star center of $\mathcal{SPR}$ is equal to $\mathcal P$.
\end{theorem}
\begin{proof}
	For the sake of brevity, we will only show that: The star center of $\mathcal{SN}$ is equal to $\mathcal P$. The proofs of the other claims are very similar.
	
	To see why $\mathcal{SN}$ is star-shaped with star center containing  $\mathcal {P}$, consider a SNTP map $\Psi$ and a positive TP map $\Phi$, which is also SN. Suppose that $\Psi(\rho)$ and $\rho$ are density matrices then for all $\lambda \in (0,1]$ we have $\lambda\Psi(\rho) +(1-\lambda)\Phi(\rho)\geq 0$. Hence $\lambda\Psi+(1-\lambda)\Phi$ is SNTP for $\lambda \in (0,1]$ but even when $\lambda=0$ we have $\Phi$ is SN by assumption. Therefore, $\mathcal {P}$ is in the star center of $\mathcal {SN}$.
	
	For the other inclusion, suppose that $\Psi$ is a SNTP map which is not positive. 
	We will show that $\Psi$ is not in the star center of $\mathcal{SN}$.
	Firstly, note that since $\Psi$ is not positive, there exists a pure state of $\ket \phi \bra \phi \in \mathcal D_{\mathcal H}$ with $\Psi(\ket \phi \bra \phi)$ indefinite. 
	
	We will now prove the following statement: there is a $\epsilon >0$, a non-zero matrix $Z\in B(\mathcal K)$ with $\Tr(Z)=0$ and a density matrix $\sigma \in \mathcal D_{\mathcal K}$ such that for all $\rho \in B_{\epsilon}(\ket \phi \bra \phi)\cap \mathcal D_{\mathcal H}$, and all $r \in \mathbb R$ we have that 
	\[
		\Psi(\rho) + \sigma + rZ \text{ is indefinite}.
	\]
	Where $B_{\epsilon}(\ket \phi \bra \phi)$ is the open ball of radius $\epsilon$ centred at $\ket \phi \bra \phi$, with respect to the norm induced by the inner product on $\mathcal H$.
	To prove this statement, we pick $\sigma \in \mathcal D_{\mathcal H}$ so that $A=\Psi(\ket \phi \bra \phi)+\sigma$ is indefinite. 
	Since $A$ is indefinite and hermitian, there are unit $\ket{v_1},\ket{v_2} \in \mathcal K$ with $\ket{v_1},\ket{v_2}$ orthogonal and   $\bra{v_1} A \ket{v_1} <0$, $\bra{v_2} A \ket{v_2} >0$. 
	Let $Z= \ket{v_2}\bra{v_1} + \ket{v_1} \bra{v_2} $ then, we have $\bra{v_1} Z \ket{v_1} =0$, $\bra{v_2} Z \ket{v_2}=0$. 
	Define $f_1(\rho)= \bra{v_1} \Psi(\rho) + \sigma \ket{v_1} $ and   $f_2(\rho)= \bra{v_2} \Psi(\rho) + \sigma \ket{v_2} $, one can see that $f_1(\rho)= \bra{v_1} \Psi(\rho) + \sigma +r Z\ket{v_1} $ and  $f_2(\rho)= \bra{v_2} \Psi(\rho) + \sigma +r Z \ket{v_2} $ for all $r\in \mathbb R$.
	Since, $f_1(\ket \phi \bra \phi) <0$, $f_2(\ket \phi \bra \phi) >0$ and both $f_1,f_2$ are continuous there exists an $\epsilon >0$ such that for all $\rho \in B_{\epsilon}(\ket \phi \bra \phi)$ we have both  $f_1(\rho) <0$ and $f_2(\rho) >0$.
	It follows that the statement holds for this $\epsilon$, $\sigma$ and $Z$. 
	
	Define, $\Phi(X)= \Tr(X)\sigma + k\Tr((I_{\mathcal H}-\ket \phi \bra \phi)X)Z$ where $k\in \mathbb R$ we will determine shortly. 
	Note that $\Phi(X)$ is SNTP. 
	Furthermore, it is not hard to see that when $X\in \mathcal D_{\mathcal H}$ that: $\Tr((I_{\mathcal H}-\ket \phi \bra \phi)X)=0$ if and only if $X=\ket \phi \bra \phi$.
	Since $\mathcal D_{\mathcal H}\setminus (B_{\epsilon}(\ket \phi \bra \phi)\cap \mathcal D_{\mathcal H})$ is compact, we have $m=\inf\{|\Tr((I_{\mathcal H}-\ket \phi \bra \phi)X)| : X \in \mathcal D_{\mathcal H}\setminus (B_{\epsilon}(\ket \phi \bra \phi)\cap \mathcal D_{\mathcal H})\}>0$. 
	We pick $k\in \mathbb R$, $k>0$, large enough so that the smallest magnitude strictly negative eigenvalue of $kmZ$ is less than $-3$. 
	This means there is a unit $\ket \psi \in \mathcal K$ with $km \bra \psi Z \ket \psi <-3$.
	
	We now show that $\frac 1 2 (\Psi + \Phi)$ is not SN, in particular we show that $\Psi(X) + \Phi(X) \not \geq 0$ for all $X \in \mathcal D_{\mathcal H}$.
	To see this, notice that when $X\in  B_{\epsilon}(\ket \phi \bra \phi)\cap \mathcal D_{\mathcal H}$ we have
	\[
	\Psi(X) + \Phi(X) =\Psi(X) + \sigma + rZ
	\]
	where $r=k \Tr((I_{\mathcal H}-\ket \phi \bra \phi)X) \in \mathbb R$. 
	By the claim, $\Psi(X) + \Phi(X)$ is indefinite, and so is not positive semi-definite.
	When $X \in \mathcal D_{\mathcal H}\setminus (B_{\epsilon}(\ket \phi \bra \phi)\cap \mathcal D_{\mathcal H})$, we pick a unit $\ket \psi \in \mathcal K$ with $\bra \psi kmZ \ket \psi <-3$ as mentioned before and consider
	\begin{align*}
		\bra \psi \Psi(X) + \Phi(X) \ket \psi &= \bra \psi \Psi(X)\ket \psi+\bra \psi\sigma \ket \psi+ k \Tr((I_{\mathcal H}-\ket \phi \bra \phi)X) \bra \psi Z \ket \psi \\
		&\leq  1+1 + km \bra \psi Z \ket \psi \\
		&\leq 2-3 <0. 
	\end{align*}
	Therefore, $\Psi(X) + \Phi(X)$ is not  positive semi-definite for all $X\in \mathcal D_{\mathcal H}$. It follows that $\frac 1 2 (\Psi + \Phi)$ is not SN; since this is a convex combination of two SN maps neither can be in the star center. Hence, $\Psi$ is not in the star center, which proves the theorem. 
\end{proof}

\begin{pro}\label{thm:convex_hull}
    The convex hull $\Conv(\mathcal{SP}(\mathcal{H},\mathcal{K})) $ of $\mathcal{SP}(\mathcal{H},\mathcal{K})$ is the set of HPTP maps. 
\end{pro}

\begin{proof}
    Let $\{H_i\}$ be a basis of $B(\mathcal{H})$ consisting entirely of Hermitian elements with both $H_1$ and $H_2$ being invertible density matrices.  Let $\rho$ be an arbitrary density matrix in $ B(\mathcal{K})$.  Now suppose $\Phi$ is an HPTP map from  $B(\mathcal{H}) \to B(\mathcal{K})$.  Now define $\Psi_1$ and $\Psi_2$ as follows:  $\Psi_1(H_1)=\Psi_2(H_2)=\rho$ (hence $\Psi_1$ and $\Psi_2$ are semipositive), $\Psi_1(H_2)=2\Phi(H_2)-\rho$, $\Psi_2(H_1)=2\Phi(H_1)-\rho$, and $\Psi_1(H_k)=\Psi_2(H_k)=\Phi(H_k)$ for all $k\geq 3$.  Since $\Phi=\frac{1}{2}(\Psi_1+\Psi_2)$, every HPTP map is the convex combination of two semi-positive maps.
\end{proof}

\begin{theorem}\label{thm:set_relations_app}
     Let $\mathcal{H}$ and $\mathcal{K}$ be finite-dimensional Hilbert spaces with dimensions greater than $1$. Let $\mathcal {HP}$ be the set of all HPTP linear maps from $B(\mathcal H)$ to $B(\mathcal K)$ and let $\norm{\cdot}_{\mathcal {HP}}$ be some norm on $\mathcal {HP}$. Then the following hold:
     \begin{enumerate}
         \item $\mathcal {SP}=\set{\Psi \in \mathcal {HP} : \Psi \text{ is SP} }$ is open and unbounded in $(\mathcal {HP}, \norm{\cdot}_{\mathcal {HP}})$ but not convex. It is star-shaped with a star center equal to $\mathcal {SP}\cap \mathcal {P}$. \label{thm:set_relations.SP_app}
         \item $\mathcal {SN}=\set{\Psi \in \mathcal {HP} : \Psi \text{ is SN} }$ is closed and unbounded in $(\mathcal {HP}, \norm{\cdot}_{\mathcal {HP}})$ but not convex. It is star-shaped with a star center equal to $ \mathcal {P}$. \label{thm:set_relations.SN_app}
         \item $\mathcal {SPR}=\set{\Psi \in \mathcal {HP} : \Psi \text{ is SPR} }$ is unbounded, not open and not closed in $(\mathcal {HP}, \norm{\cdot}_{\mathcal {HP}})$ nor is it convex. It is star-shaped with a star center equal to $ \mathcal {P}$. \label{thm:set_relations.SPR_app}
         \item $\mathcal {P}=\set{\Psi \in \mathcal {HP} : \Psi \text{ is positive} }$ is compact in $(\mathcal {HP}, \norm{\cdot}_{\mathcal {HP}})$ and convex. \label{thm:set_relations.P_app}
         \item $\mathcal {CP}=\set{\Psi \in \mathcal {HP} : \Psi \text{ is CP} }$ is compact in $(\mathcal {HP}, \norm{\cdot}_{\mathcal {HP}})$ and convex. \label{thm:set_relations.CP_app}
         \item 
         \[
            \mathcal {CP}\subseteq \mathcal {P} \subseteq \mathcal{SPR} \subseteq \mathcal {SN}
         \]
         and
         \[
         \mathcal {SP} \subseteq \mathcal{SPR} 
         \]
         \label{thm:set_relations.inclusions_app} 
         \item 
         \[
         \begin{array}{lr}
           \operatorname{int}(\mathcal {SN})= \mathcal {SP}   &  \mathcal{SN}=\overline{\mathcal {SP}}
         \end{array}
         \] \label{thm:set_relations.SN_regular_app}
     \end{enumerate}
\end{theorem}
\begin{proof}
    For \ref{thm:set_relations.SP_app}, we first show that $\mathcal{SP}$ is unbounded. Consider the maps
    \begin{align*}
        \Psi_{TP}(X)&= \Tr(X)\sigma \\
        \Psi_{TA}(X)&= \Tr(Z_1X)Z_2 
    \end{align*}
    where $X,\sigma,Z_1 \in B(\mathcal H)$, $Z_2 \in B(\mathcal K)$, $\sigma$ is an invertible density matrix and $Z_1,Z_2$ are non-zero trace zero hermitian matrices. 
    Now consider for $k\in \mathbb N$ $\Psi_k= \Psi_{TP} +k \Psi_{TA}$, for every $k\in \mathbb N$ $\Psi_k$ is TP (as $Z_2$ has zero trace) and it is SP since $\Psi_k(\rho)=\sigma$ where $\rho = \frac{1}{\dim(\mathcal H)}I_{\mathcal H}$. Note that $\norm{\Psi_{TA}}_{\mathcal {HP}}\neq 0$ as $\Psi_{TA}\neq 0$ and we see
    \[
        \norm{\Psi_{k}}_{\mathcal {HP}} \geq k \norm{\Psi_{TA}}_{\mathcal {HP}} - \norm{\Psi_{TP}}_{\mathcal {HP}}
    \]
    which is unbounded as $k\to \infty$. Therefore, $\mathcal{SP}$ is unbounded.

    Next, we argue that $\mathcal{SP}$ is open. Let $\Psi$ be SPTP and let $\Phi$ be any HPTP map with $\norm{\Phi}_{\mathcal {HP}}=1$. Suppose that $\Psi(\rho)$ and $\rho$ are invertible density matrices then there an $\epsilon>0$ such that $\Psi(\rho) + \epsilon \Phi(\rho) >0$ which means $\Psi +\epsilon \Phi$ is SP. From here, it can be seen that $\mathcal{SP}$ is open.

    To see why $\mathcal{SP}$ is start shaped with a star center equal to $\mathcal {SP}\cap \mathcal {P}$, we refer to \cref{thm:star}.
    
    To see why $\mathcal{SP}$ is not convex, we refer to Proposition~\ref{thm:convex_hull}. In Proposition~\ref{thm:convex_hull}, we show that $\Conv(\mathcal SP)\neq \mathcal {SP}$ ($\Conv(M)$ denotes the convex hull of a set $M$), which means that $\mathcal {SP}$ is not convex.

    The proof of \ref{thm:set_relations.SN_app} follows quickly from the arguments given in \ref{thm:set_relations.SP_app} expect for closedness of $\mathcal{SN}$. 
    To see why $\mathcal{SN}$ is closed, suppose that $\{\Psi_n\}_{n\in \mathbb N}$ is a sequence in $\mathcal{SN}$ which converges to $\Psi$ in $\mathcal {HP}$; note that this sequence also converges uniformly. 
    By SN, for each $n\in \mathbb N$, there is a density matrix $\rho_n$ with $\Psi_n(\rho_n)\geq 0$; since the set of density matrices is compact, there is a subsequence $\{\rho_{n_k}\}_{k\in \mathbb N}$ which converges to a density matrix $\rho$. 
    We also have that $\{\Psi_{n_k}\}_{k\in \mathbb N}$ converges to $\Psi$ in $\mathcal {HP}$, since it also converges uniformly, we have $\lim_{k\to\infty}\Psi_{n_k}(\rho_{n_k})= \Psi(\rho)$ and since the set of positive semi-definite matrices is closed we also have  $\Psi(\rho)\geq 0$. 
    This means $\Psi$ is SN and $\mathcal{SN}$ is closed. 

    Again the proof of \ref{thm:set_relations.SPR_app} follows quickly from the arguments given in \ref{thm:set_relations.SP_app} expect the non-openness and non-closedness. To see the non-closedness, note that by \ref{thm:set_relations.SN_regular_app} the set $\mathcal {SN}$ is closed and cannot have any isolated points, and since $\mathcal{SPR}\subseteq \mathcal{SN}$ we need only find a map in $\mathcal{SN}$ not in $\mathcal{SPR}$. Consider the map
    \[
        \Psi(X)= \Tr(X)E_{11}^{\mathcal K} + \Tr(E_{22}^{\mathcal H}X)(E_{11}^{\mathcal K} -E_{22}^{\mathcal K})
    \]
    (which is well defined as $\dim(\mathcal H),\dim(\mathcal K)>1$) where $E_{ij}^{\mathcal H},E_{ij}^{\mathcal K}$ are the choi basis of $\mathcal H$ and $\mathcal K$. Then, $\Psi$ is TP and $B(\mathcal K_\Psi) \simeq \mathbb C ^{2\times 2}$. 
    $\Psi$ is SN as $\Psi(E_{11}^{\mathcal H})=E_{11}^{\mathcal K}$  but as a map from $B(\mathcal H)$ to $B(\mathcal K_\Psi)$ it is not SP. Since if $X>0$ then $\Tr(E_{22}^{\mathcal H}X)>0$ which means $\Tr(E_{22}^{\mathcal K}\Psi(X)) )<0$, thus $\Psi(X)\not \geq 0$. Therefore, $\Psi$ is not SPR. 

    To show that $\mathcal{SPR}$ is not open, we show that $\mathcal{SPR}$ contains a boundary point. The map $\Psi(X)= \Tr(X)\sigma$ where $\sigma$ is a pure state and in $B(\mathcal K_\Psi)$ we have $\sigma >0$ so $\Psi$ is SPR. 
    However, as a map from $B(\mathcal H)$ to $B(\mathcal K)$ we have that $\lim_{\epsilon \to 0} \Gamma_\epsilon = \Phi$ and $\lim_{\epsilon \to 0} \Phi_\epsilon = \Phi$ where
    \begin{align*}
        \Gamma_\epsilon(X) &= \Psi(X)+\epsilon \Tr(X)Z \\
        \Phi_\epsilon(X) &= (1-\epsilon)\Psi(X)+\epsilon \Tr(X) I_{\mathcal K},
    \end{align*}
    $X \in B(\mathcal H)$, $Z \in B(\mathcal K)$, $Z$ is an invertible trace zero matrix with $\bra \phi Z \ket \phi <0$ and $\sigma \ket \phi =0$ for some unit $\ket \phi \in \mathcal K$.
    Then, for all $\epsilon \in (0,1)$ we have that $\Gamma_\epsilon$ is not SN since $\bra \phi \Gamma_\epsilon(X) \ket \phi=\bra \phi \epsilon Z \ket \phi <0$ for $X$ being a density matrix. Also, for all $\epsilon \in (0,1)$, we have that $\Phi_\epsilon$ is SPTP and positive (and therefore SPR). This shows that $\Psi$ lies on the boundary of $\mathcal{SPR}$, so $\mathcal{SPR}$ is not open. 

    \Cref{thm:set_relations.P_app,thm:set_relations.CP_app} are well known. 

    We now prove~\cref{thm:set_relations.inclusions_app}. It is well known that  $\mathcal{CP}\subseteq \mathcal P$. Let $\Phi \in \mathcal P$ then, as a map from $B(\mathcal H)$ to $B(\mathcal K_\Phi)$ is SP, since by \cref{thm:range_reduection_app} there is an invertible density matrix $\rho \in B(\mathcal H)$ such that $\mathcal{K}_\Phi= \operatorname{Range} (\Phi(\rho))$ and 
    so $\Phi(\rho)$ is positive and invertible in $B(\mathcal{K}_\Phi)$. Hence $\Phi$ is SPR.

    The inclusion $\mathcal {SPR}\subseteq \mathcal{SN}$ follows from the fact: if $ \sigma \in B(\mathcal{K}_\Psi)$ is an invertible positive matrix then, in $B(\mathcal{K})$ $\sigma$ is positive semi-definite. 

    Finally, the inclusion $\mathcal{SP} \subseteq \mathcal{SPR}$ is immediate since if $\Psi(\rho)>0$ (in $B(\mathcal K)$) for $\rho>0$ then, $\operatorname{Range}(\Psi(\rho))=\mathcal K$. Hence, $\mathcal K_{\Psi}=\mathcal K$.

    To prove, \ref{thm:set_relations.SN_regular_app}, we first consider $\operatorname{int}(\mathcal {SN})= \mathcal {SP}$. We already know that $\operatorname{int}(\mathcal {SN})\supseteq \mathcal {SP}$ since $SP$ is open and contained in $\mathcal {SN}$. 
    So suppose that $\Psi \in \operatorname{int}(\mathcal {SN})$ then for some $\epsilon>0$ we have $(1-\epsilon)\Psi +\epsilon \Phi \in \mathcal{SN}$ where $\Phi(x)= \Tr(x)b$ is SNTP when $b$ is an invertible density matrix. 
    Let $\rho \geq 0$ have $\Psi(\rho) \geq 0$ then for $a$, an invertible density matrix in $\mathcal H$, we have $\Psi((1-\epsilon)\rho +\epsilon a)=(1-\epsilon)\Psi(\rho) +\epsilon \Phi(a)>0$. Since $(1-\epsilon)\rho +\epsilon a$ and $\Psi((1-\epsilon)\rho +\epsilon a)$ are invertible density matrices we know $\Psi$ is SP. This proves the equality. 

    On the other hand, we know $\mathcal{SN}\supseteq \overline{\mathcal{SP}}$ is true since $\mathcal{SN}\supseteq {\mathcal{SP}}$  and $\mathcal{SN}$ is closed. 
     Let $\Psi$ be a SN map then consider for $\lambda \in [0,1)$ the map $\Gamma_\lambda =(1-\lambda)\Psi +\lambda \Phi$ where $\Phi(x)= \Tr(x)I_{\mathcal K}$. 
     We claim that for  $\lambda \in [0,1)$  $\Gamma_\lambda$ is SP, to see this let $\rho \geq 0$ be in $\mathcal H$ with $\Psi(\rho) \geq 0$ then  $\Gamma_\lambda(\rho)= \lambda \Psi(\rho)  +(1-\lambda)I_{\mathcal K} >0$ for $\lambda \in [0,1)$ and so $\Gamma_\lambda$ for $\lambda \in [0,1)$ is SP by \cref{thm:SP_noinv}. 
     Also note that  for $\lambda_n= 1-\frac 1 n$, $n\in \mathbb N$ we have that $\lim_{n\to \infty}\Gamma_{\lambda_n}=\Psi$ and $\Psi$ is the limit of SP maps, so $\Psi \in \overline{\mathcal{SP}}$ as required. 
\end{proof}

\begin{theorem} 
Let $\mathcal{H}$ and $\mathcal{K}$ be finite dimensional Hilbert spaces and let $\Psi$ be an HPTP map from $B(\mathcal{H}) \to B(\mathcal{K})$.  Then one and only one of the following statements is true.
 \begin{enumerate}
   \item $\Psi$ is a semi-positive map.
   \item $-\Psi^*$ is a semi-nonegative map.
  \end{enumerate}
\end{theorem}

\begin{proof}  Let $\mathcal{S}_\mathcal{H}$  and $\mathcal{S}_\mathcal{K}$ be the set of invertible density matrices in  $B(\mathcal{H})$ and $B(\mathcal{K})$ respectively.  If $\Psi$ is not semi-positive, then $\Psi(\mathcal{S}_\mathcal{H} )$ and $\mathcal{S}_\mathcal{K}$ are disjoint convex set; therefore, there exists a hyperplane that separates these two sets.  We can choose this separating hyperplane to also be a supporting hyperplane of $\mathcal{S}_\mathcal{K}$ which means it is of the form $\{X: \tr(\rho X)=0\}$ for some singular density matrix $\rho$.  Then $\tr(\Psi^*(\rho)X)=\tr(\rho \Psi (X))\le 0$ for all positive semi-definite $X\in B(\mathcal{H})$.  Therefore $\Psi^*(\rho)$ is negative definite and hence $-\Psi^*$ is a semi-nonegative map.

A similar argument shows that if $(-\Psi^*)$ is not a semi-nonnegative map, then $\Psi$ is a semi-positive map. 

\end{proof}

\section{Representations of HPTP Maps}

Similarly to CPTP maps, we can define the Choi representation of HPTP maps. For an HPTP map $\Psi$ acting on L($\mathcal{H}$), the Choi representation $J\lrp{\Psi}$ of $\Psi$ is given by 
\[
J\lrp{\Psi}=\lrp{\Phi\otimes I_{\mathcal{H}}}\lrp{\veco\lrp{I_{\mathcal{H}}}\veco\lrp{I_{\mathcal{H}}}^{\dagger}},
\]
where L$(\mathcal{H})$ is the space of linear operators acting on the Hilbert space $\mathcal{H}$ and $I$ is the identity matrix.

If $\Psi$ is HPTP, then the Choi representation $J\lrp{\Psi}$ is hermitian \cite[Theorem 2.25]{watrous2018theory}. The Choi representation thus admits a Jordan-Hahn decomposition
$J\lrp{\Psi}=P-Q$
where $P,Q$ are positive semi-definite. Let
\[J\lrp{\Phi_0}=\frac{1}{\Tr P}P,\quad J\lrp{\Phi_1}=\frac{1}{\Tr Q}Q\]
We see that $J\lrp{\Phi_0}$ and $J\lrp{\Phi_1}$ are positive semi-definite operators with trace 1. They are the Choi representations of CPTP maps $\Phi_0$, $\Phi_1$ respectively \cite[Corollary 2.27]{watrous2018theory}. 
This decomposition ensures us the operator sum representation used in Section 4 in the main text.
By linearity of the Choi representation, we have that $\Psi=p_0\Phi_0-p_1\Phi_1$ where $p_0-p_1=\Tr P-\Tr Q=1$. Therefore, the operator-sum representation of $\Psi$ is the weighted difference of the operator-sum representations of CPTP maps $\Phi_0$ and $\Phi_1$. Let $\lrcb{A_k}_{k\in K}$, $\lrcb{B_j}_{j\in J}$ be the operators for the operator-sum representation of $\Phi_0$ and $\Phi_1$ respectively. Then the operator-sum representation for $\Psi$ is given by $\Psi(X)=p_0\sum_kA_kXA_k^{\dagger}-p_1\sum_j B_jXB_j^{\dagger}$.
We, therefore, refer to the operator-sum representation of a given HPTP map as $\Psi : \lrcb{\text{sign}(i),E_i}$ where $\Psi(X)=\sum_i\text{sign}(i)E_iXE_i^{\dagger}$.

Similarly to CPTP maps, the operator-sum representation for HPTP maps is not unique. The unitary equivalence of HPTP maps is derived in a similar manner to the unitary equivalence of CPTP maps. Due to the following theorems, we can diagonalize the Hermitian matrix $\alpha$ in Theorem 6 in the main text. 


\begin{theorem}\label{thm:hptp_unitary_equivalence}
    Let $\Psi$, $\Xi$ be HPTP maps. If there exists sets of linear operators $\lrcb{A_k}_{k\in K}$, $\lrcb{B_j}_{j\in J}$, $\lrcb{C_i}_{i\in \mathcal{I}}$ and $\lrcb{D_l}_{l\in L}$ such that
    \[\Psi(X)=\sum_k A_kXA_k^{\dagger}-\sum_jB_jXB_j^{\dagger},\quad \Xi(X)=\sum_iC_iXC_i^{\dagger}-\sum_{l}D_lXD_l^{\dagger}\]
    with $\left|K\right|=\left|L\right|=\left|\mathcal{I}\right|=\left|J\right|$. Then $\Psi=\Xi$ if and only if there exists a unitary operator $U$ with rows labelled by $K\cup L$ and columns labelled by $\mathcal{I}\cup J$ such that 
    \[A_k=\sum_{i}U(k,i)C_i+\sum_{j}U(k,j)B_j,\quad D_l=\sum_{i}U(l,i)C_i+\sum_{j}U(l,j)B_j\]
\end{theorem}

\begin{proof}
    We consider first the $\Leftarrow$ direction.
    Consider
    $\sum_kA_kXA_k^{\dagger}+\sum_lD_lXD_l^{\dagger}$.

    It can be shown by direct calculation after substituting $A_k,D_l$ with their expansions over $C_i,B_j$ that this matrix is equivalent to 
    $\sum_jB_jXB_{j}^{\dagger}+\sum_iC_iXC_{i}^{\dagger}$
    We can rearrange the terms to obtain the equality :
    $\sum_kA_kXA_k^{\dagger}-\sum_jB_jXB_j^{\dagger}=\sum_iC_iXC_i{\dagger}-\sum_lD_lXD_l^{\dagger}$,
    implying $\Phi=\Xi$.

    We now consider the $\Rightarrow$ direction. Suppose that 
    \[\Psi(X)=\sum_kA_kXA_k^{\dagger}-\sum_jB_jXB_j^{\dagger}=\sum_iC_iXC_i^{\dagger}-\sum_lD_lXD_l^{\dagger}=\Xi(X)\]
    
    for all $X$. Their Choi representations must agree. The Choi representations are given by \cite[Proposition 2.20]{watrous2018theory}. 

    \[J\lrp{\Psi}=\sum_k\veco\lrp{A_k}\veco\lrp{A_k}^{\dagger}-\sum_j\veco\lrp{B_j}\veco\lrp{B_j}^{\dagger}=\sum_i\veco\lrp{C_i}\veco\lrp{C_i}^{\dagger}-\sum_l\veco\lrp{D_l}\veco\lrp{D_j}^{\dagger}\]
    
    where $\veco$ denotes the vector operator correspondence \cite{watrous2018theory}. This equation is equivalent to
    \[\sum_k\veco\lrp{A_k}\veco\lrp{A_k}^{\dagger}+\sum_l\veco\lrp{D_l}\veco\lrp{D_j}^{\dagger}=\sum_j\veco\lrp{B_j}\veco\lrp{B_j}^{\dagger}+\sum_i\veco\lrp{C_i}\veco\lrp{C_i}^{\dagger}\]
    
    We now replicate the proof of \cite[Corollary 2.23]{watrous2018theory}. We define an orthonormal basis $\lrcb{e_a}_{a\in K\cup L\cup \mathcal{I}\cup J}$ over a Hilbert space $\mathcal{Z}$. We define the following vectors :
    \[\ket{u}=\sum_k\veco\lrp{A_k}\otimes e_k+\sum_l\veco\lrp{D_l}\otimes e_l,\quad \ket{v}=\sum_j\veco\lrp{B_j}\otimes e_j+\sum_i\veco\lrp{C_i}\otimes e_i\]
    These vectors lead to an equivalence of purifications
    $\tr_{\mathcal{Z}}\lrp{\ket{u}\bra{u}}=\tr_{\mathcal{Z}}\lrp{\ket{v}\bra{v}}$.
    We can invoke the unitary equivalence of purifications \cite[Theorem 2.12]{watrous2018theory}. The argument is identical to that of \cite[Corollary 2.23]{watrous2018theory}. There exists a unitary operator $U$ indexed by $\lrp{K\cup L}\times\lrp{\mathcal{I}\cup J}$ such that
    \[\veco\lrp{A_k}=\sum_{i}U(k,i)\veco\lrp{C_i}+\sum_jU(k,j)\veco\lrp{B_j},\quad \veco\lrp{D_l}=\sum_{l}U(l,i)\veco\lrp{C_i}+\sum_jU(l,j)\veco\lrp{B_j}\]
    which is equivalent to the desired result by linearity of the vector-operator correspondence.
    
\end{proof}

\begin{theorem}

Let $\Psi$, $\Xi$ be HPTP maps. Let $\Psi=p_0\Phi_0-p_1\Phi_1$ and $\Xi=c_0\chi_0-c_1\chi_1$, where $\Phi_0,\Phi_1,\chi_0,\chi_1$ are CPTP maps. Let $\rank\lrp{J\lrp{\Psi}}=\rank\lrp{J\lrp{\Phi_0}}+\rank\lrp{J\lrp{\Phi_1}}$ and $\rank\lrp{J\lrp{\Xi}}=\rank\lrp{J\lrp{\chi_0}}+\rank\lrp{J\lrp{\chi_1}}$. Let $\lrcb{A_k}_{k\in K}$, $\lrcb{B_j}_{j\in J}$, $\lrcb{C_i}_{i\in \mathcal{I}}$ and $\lrcb{D_l}_{l\in L}$ be sets of linear operators such that 
    \[\Psi(X)=\sum_kA_kXA_k^{\dagger}-\sum_jB_jXB_j^{\dagger},\quad \Xi(X)=\sum_iC_iXC_i^{\dagger}-\sum_{l}D_lXD_l^{\dagger}\]
where $\left|K\right|=\rank\lrp{J\lrp{\Phi_0}}$, $\left|J\right|=\rank\lrp{J\lrp{\Phi_1}}$, $\left|\mathcal{I}\right|=\rank\lrp{J\lrp{\chi_0}}$ and $\left|L\right|=\rank\lrp{J\lrp{\chi_1}}$. Then $\Psi=\Xi$ if and only if there exists unitary operators $U_1$ with rows labelled by $K$ and columns labelled by $\mathcal{I}$ and $U_2$ with rows labelled by $J$ and columns labelled by $L$ such that 
    \[A_k=\sum_iU_1(k,i)C_i,\quad D_l=\sum_j U_2(l,j)B_j\]

\end{theorem}

\begin{proof}
    The proof for the $\Leftarrow$ direction follows directly from the unitary equivalence of Kraus representations for $CPTP$ maps. It can be proven by direct calculation.

   We now consider the $\Rightarrow$ direction. First, we consider the Choi representation for $\Psi$.
   \[
   J\lrp{\Psi}=\sum_k\veco\lrp{A_k}\veco\lrp{A_k}^{\dagger}-\sum_j\veco\lrp{B_j}\veco\lrp{B_j}^{\dagger}\]
   We consider the following operators :
   \[\Tilde{A}=\sum_k\veco\lrp{A_k}\veco\lrp{A_k}^{\dagger},\quad \Tilde{B}=\sum_j\veco\lrp{B_j}\veco\lrp{B_j}^{\dagger}
   \]
   These operators are positive, and therefore admit spectral decompositions. The operators decompose into $\rank\lrp{C(\Phi_0)}$ and $\rank\lrp{C(\Phi_1)}$ rank 1 projection operators respectively, since they are sums of that many rank 1 operators and must have a total rank of $\rank\lrp{C(\Phi_0)}+\rank\lrp{C(\Phi_1)}$. 
   Let $\lambda^{(A)}_k$, $\lambda^{(B)}_j$ and $\Pi^{(A)}_k$, $\Pi^{(B)}_j$ denote the eigenvalues and rank 1 projectors of the decomposition for $\Tilde{A}$ and $\Tilde{B}$ respectively. Since the operators are positive, $\lambda^{(A)}_k,\lambda^{(B)}_j>0$. We let eigenvalues repeat to constrain the rank of our projectors to 1 in the spectral decomposition.

   Further, $C(\Psi)$ is Hermitian and therefore admits a spectral decomposition into $\rank\lrp{C(\Psi)}=\rank\lrp{C(\Phi_0)}+\rank\lrp{C(\Phi_1)}$ rank 1 projectors. 
   Let $\lambda_n^{(\Psi)}$, $\Pi_n^{(\Psi)}$ denote the eigenvalues and projection operators of the decomposition. We obtain the equation
   \[\sum_{n=1}^{\rank(J\lrp{\Psi})}\lambda_n^{(\Psi)}\Pi_n^{(\Psi)}=\sum_{k=1}^{\rank(J\lrp{\Phi_0})}\lambda_k^{(A)}\Pi^{(A)}_{k}-\sum_{j=1}^{\rank(J\lrp{\Phi_1})}\lambda_j^{(B)}\Pi^{(B)}_j\]
   The spectral decomposition is unique up to ordering, therefore we can identify the projectors on the left-hand side of the equation to the projectors on the right-hand side of the equation. In particular, we have that the projection operators are orthogonal and $\Pi^{(A)}_k\Pi^{(B)}_j=0$ for all $k\in K,\ j\in J$. Further, since $J(\Psi)=J(\Xi)$, we can decompose $J(\Xi)$ in a similar manner to obtain the equation :
   \[
   \sum_{k=1}^{\rank(J\lrp{\Phi_0})}\lambda_k^{(A)}\Pi^{(A)}_{k}-\sum_{j=1}^{\rank(J\lrp{\Phi_1})}\lambda_j^{(B)}\Pi^{(B)}_j=\sum_{i=1}^{\rank(J\lrp{\chi_0})}\lambda_i^{(C)}\Pi^{(C)}_i-\sum_{d=1}^{\rank(J\lrp{\chi_1})}\lambda_l^{(D)}\Pi^{(D)}_l
   \]
   By identifying the positive and negative eigenvalue eigenspaces on both sides of the equation, we obtain $\Pi^{(C)}_i\Pi_j^{(B)}=\Pi_k^{(A)}\Pi_l^{(D)}=0$ for all $i\in \mathcal{I},\ j\in J,k\in K,\ l\in L$.
   
    Therefore,
    \[ \veco\lrp{A_k}^{\dagger}\veco\lrp{B_j}=\lrp{\lrp{\sum_{k'}\Pi^{(A)}_{k'}}\veco\lrp{A_k}}^{\dagger}\lrp{\sum_{j'}\Pi_{j'}^{(B)}}\veco\lrp{B_j}=0.
    \]
    And since the vector-operator correspondence is an isometry, we have that
    $\Tr\lrp{B_j^{\dagger}A_k}=0$
    for all $k\in K$, $j\in J$.
    An identical argument can be applied to obtain the relation
    $\Tr\lrp{C_i^{\dagger}D_l}=\Tr\lrp{B_i^{\dagger}C_i}=0$.
    
    Suppose $\Psi=\Xi$. We know from \Cref{thm:hptp_unitary_equivalence} that there exists a unitary such that
    \[
    A_k=\sum_{i}U(k,i)C_i+\sum_{j}U(k,j)B_j,\quad D_l=\sum_{i}U(l,i)C_i+\sum_{j}U(l,j)B_j.
    \]
    The $B_j$ must be linearly independent since $\rank\lrp{C(\Phi_1)}$ different operators that must span a space of the same dimension. Consider $\hat{B}_j$ such that
    $\Tr\lrp{\hat{B_j}^{\dagger}B_{j'}}=c_j\delta_{jj'}$ and $ \sum_{j'}\Pi^{(B)_{j'}}\hat{B_j}=B_j$
    where $c_j\in\mathbb{C}$, $c\neq 0$ is a constant. This operator must exist by linear independence. Now consider
    \[\Tr\lrp{\hat{B}_{j'}^{\dagger}A_k}=\sum_iU(k,i)\Tr\lrp{\hat{B}_{j'}^{\dagger}C_i}+\sum_jU(k,j)\Tr\lrp{\hat{B}_{j'}^{\dagger}B_j}\]
    which is equivalent to
    $0=U(k,j)c_j$.
    The calculation for $U(l,i)$ is the same. Since the constants $c_j$ are non-zero, we obtain the following conditions on the unitary $U$ : 
    $U(k,j)=0$ and $U(l,i)=0$
    for all $k\in K$, $j\in J$, $l\in L$ and $i\in \mathcal{I}$.
    Identifying $U_1(k,i)=U(k,i)$, $U_2(l,j)=U(l,j)$ leads to the desired equivalence.
    
\end{proof}

\section{Semidefinite Programming for Determine SN and SP}
Whether a linear map $\Phi$ is Hermitian-preserving or completely positive can be easily checked by writing out the Choi representation $J(\Phi)$ of $\Phi$. If $J(\Phi)$ is a Hermitian matrix, the map $\Phi$ is HP. Similarly, if $J(\Phi)$ is positive, $\Phi$ is CP. 
However, the positivity of the Choi representation does not differentiate between a non-CP SN map from a non-CP HP map. 
Here we provide a semidefinite program to determine the semi-nonnegativity of a HPTP map, $\Phi:{\mathbb C}^{n\times n} \to {\mathbb C}^{m\times m}$.

\begin{equation*}
\begin{aligned}
\min_{(y,x)\in \mathbb R \times {\mathbb R}^{n^2}} \quad & y\\
\textrm{s.t.} \quad & yI+\sum_{k=1}^{n^2} x_k(V_k \oplus \Phi(V_k))\geq 0\\
& \sum_{k=1}^{n^2} x_k \Tr(V_k) =1
\end{aligned}
\end{equation*}
where $I$ is the identify matrix on $\mathbb C ^{(n+m)\times (n+m)}$, $\set{V_k}_{k=1}^{n^2}$ is a hermitian and orthonormal basis for ${\mathbb C}^{n\times n}$, and $\oplus$ is the direct sum of matrices. 
Let  $y\in \mathbb R$, $x\in{\mathbb R}^{n^2}$ and $X=\sum_{k=1}^{n^2} x_kV_k$ with $\tr(X)=1$, suppose that the constraints of the program are satisfied for $(y,x)$ then, it is not hard to see that: $X\geq 0$ and $\Phi(X)\geq0$ if and only if $y\leq 0$. Even better, $X> 0$ and $\Phi(X)>0$ if and only if $y<0$.
Therefore, if this program terminates with global minimum $y^*>0$ then we know that the map $\Phi$ is not SN. 
If the program terminates with $y^*=0$, then we know that $\Phi$ is SN.
And if the program terminates with $y^*<0$, then we know that $\Phi$ is SP.

To determine SPR (but not SP or SN), the user would need to notice that $\Phi(X)$ is always rank deficient than reduce the dimension of $\mathbb C^{m\times m}$ to $\mathbb C^{(m-k) \times (m-k)}$ for some $k\in \mathbb N$ in a way so that $\Phi: {\mathbb C}^{n\times n} \to \mathbb C^{(m-k) \times (m-k)}$ and $\Phi(X)$ is invertible for some $X\in {\mathbb C}^{n\times n}$.

\end{document}